\documentclass[twoside]{article}

\usepackage[accepted]{aistats2024}
\usepackage[round]{natbib}

\usepackage{amsmath,amssymb, amsthm, bm, bbm, enumitem, xcolor,graphicx, subcaption,mathabx}
\usepackage{algorithmicx, algorithm, algpseudocode}
\newcommand{\one}{\mathbbm{1}}
\newcommand{\bsigma}{\bm{\sigma}}
\newcommand{\bmJ}{\bm{J}}
\renewcommand{\P}{\mathbb{P}}
\newcommand{\R}{\mathbb{R}}
\newcommand{\E}{\mathbb{E}}
\newcommand{\sech}{\mathrm{sech}}
\newcommand{\p}{\texttt{priv}}
\newcommand{\vep}{\varepsilon}

\newtheorem{theorem}{Theorem}[section]

\newtheorem{corollary}{Corollary}[section]

\newtheorem{lemma}{Lemma}[section]

\theoremstyle{definition} 
\newtheorem{definition}{Definition}[section]

\begin{document}

%

%

\twocolumn[

\aistatstitle{PrIsing: Privacy-Preserving Peer Effect Estimation via Ising Model}

\aistatsauthor{ Abhinav Chakraborty \And Anirban Chatterjee \And  Abhinandan Dalal}

\aistatsaddress{ University of Pennsylvania \And  University of Pennsylvania \And University of Pennsylvania } ]

\begin{abstract}
    The Ising model, originally developed as a spin-glass model for ferromagnetic elements, has gained popularity as a network-based model for capturing dependencies in agents' outputs. Its increasing adoption in healthcare and the social sciences has raised privacy concerns regarding the confidentiality of agents' responses. In this paper, we present a novel $(\varepsilon,\delta)$-differentially private algorithm specifically designed to protect the privacy of individual agents' outcomes. Our algorithm allows for precise estimation of the natural parameter using a single network through an objective perturbation technique. Furthermore, we establish regret bounds for this algorithm and assess its performance on synthetic datasets and two real-world networks: one involving HIV status in a social network and the other concerning the political leaning of online blogs.
\end{abstract}

\section{Introduction} \label{introduction}
The ubiquity of data available on interactions between agents in a system has led to several network models being developed to better understand and contemplate agents' responses in an interconnected environment. One such popular model is the Ising spin glass model, which was originally developed in statistical physics to model ferromagnetism. However, it has now gained popularity in applications in several social science domains, due to it's ease of interpretation and widespread applicability. Thomas Schelling's Ising-like model (\cite{schelling1971dynamic}) to explain racial segregation in US cities has become a standard practice in explaining urban dynamics (\cite{fossett2006ethnic}). \cite{stauffer2008social} also discusses how the Ising model can be used to understand language dynamics and the adoption of linguistic features from different languages without external forces, a line of work pioneered in \cite{nettle1999rate}, and later strongly reflected in future literature.

One of the interesting properties of Ising model (Equation (\ref{pmf}), discussed in detail in Section \ref{formulation}), which is a joint distribution on the outcomes of the nodes $\bm{\sigma}=(\sigma_1,\cdots, \sigma_n)\in \{\pm 1\}^n$ given an arbitrarily encoded symmetric network information matrix $\bm{J}_n:= ((\bm J_n(i,j)))$, is that it is easy to infer the influence of the neighboring nodes on the outcome of an individual node. This can be seen from the conditional probability of $\sigma_i = 1$ given the other node outcomes $\bm{\sigma}_{-i}:=(\sigma_{1},\ldots,\sigma_{i-1},\sigma_{i+1},\ldots,\sigma_{n})$,
\begin{align*} \P(\sigma_i = 1|\bm\sigma_{-i}) &= \dfrac{e^{\beta\sum_{j:j\ne i}\sigma_j \bm J_n(i,j)}}{e^{\beta\sum_{j:j\ne i}\sigma_j \bm J_n(i,j)} + e^{-\beta\sum_{j:j\ne i}\sigma_j \bm J_n(i,j)}},
\end{align*}
which increases or decreases in $\beta$ as $\sum_{j\neq i} \sigma_j \bm J_n(i,j)$, is positive or negative, respectively. This parameter $\beta$, often referred to as the inverse temperature in the physics literature, encapsulates the extent of influence of neighbors in the network (see for example, \cite{daskalakis2020logistic}).

However, the applicability of Ising model in social networks comes with its concern in privacy. In fact, \cite{abawajy2016privacy} and \cite{zhou2008brief} discuss a multitude of privacy preservation techniques when presenting network data, particularly with the boom of current network data. The concerns of privacy in social network analysis has indeed been a concern echoed by many (\cite{backstrom2007wherefore, srivastava2008data}), for instance, a powerful adversary with access to others' data might be able to conclude one's outcome from a non-private algorithm, particularly since the outcomes in a network are dependent (\cite{liu2016dependence}). In fact, Ising model in itself has been or can be used to study several sensitive or potentially sensitive data on:
\vspace*{-0.5em}
\begin{itemize}
    \setlength\itemsep{0em}
    \item \textbf{Transmission of contagious diseases:} For instance \cite{mello2021epidemics} study epidemic transmission concepts from Covid-19 using Ising model.
    \item \textbf{Tax evasion dynamics} and peers' influence on such behavior, as studied by \cite{zaklan2009analysing} using Ising model, and further enriched by \cite{pickhardt2014income}, \cite{giraldo2021tax}.
    \item \textbf{Enforcing social behavior} as discussed by \cite{cajueiro2011enforcing} using Ising model for modeling harmful behaviors, like smoking decisions (\cite{krauth2006social}), criminal behavior (\cite{glaeser1996crime}), investment decisions (\cite{duflo2002participation}), etc.
    \item \textbf{Sexually transmitted diseases}, like that studied by \cite{potterat2002risk}, for HIV transmission in a network based study from Colorado Springs. They also incorporate several sensitive information like drug injection usage of agents involved.
\end{itemize}
Such applications motivate the need for privacy-preserving techniques for analysis. Indeed, quite a few of the papers cited study the model through simulations, as such data is hard to collect and are often unreliable due to the potential untruthful reporting for privacy concerns. From eavesdropping medical and financial agencies, to incriminating evidence, social taboos and voting freedoms; these examples show why privacy is of utmost importance in studying these behaviors, so that truthful data collection can be incentivized and valid inferences can be drawn while ensuring the individuals' privacy.

\subsection{Related Works}

There has been a growing literature for theoretical analysis of the Ising model, and advances have been made in understanding the non-standard estimation techniques in regard to the same. \cite{chatterjee2007estimation} is one of the pioneers in this literature, where he shows the $\sqrt n$ consistency of the maximum pseudo-likelihood estimator. On the other hand, \cite{bhattacharya2018inference} extends this result to $\sqrt{a_n}$-consistency based on conditions of the log partition function, thus completing the result of consistency for all the regimes. We build on these previous works to incorporate the non-statistical constraint of differential privacy, and quantify the loss of efficiency due to the privacy requirement. \cite{mukherjee2022testing} also analyze the difficulty in the estimation of the parameter of the Ising model in certain regimes, and draws parallels with the joint estimation strategies demonstrated in \cite{ghosal2020joint}. Theoretical explorations into distribution testing with Ising Models have been studied in  \cite{daskalakis2019testing}.

In this work, we use techniques from \cite{kifer2012private}. They however deal with independent data structures, which is in stark contrast with the dependent structure of that of the Ising model, thus requiring the necessity for developing new arguments and drawing insights from the Ising literature to prove error bounds on the differentially private estimator.

However, it must be noted that the notion of outcome-differential privacy is different from the usual edge-differential privacy (eg: \cite{mohamed2022differentially}, \cite{chen2023private}, etc.) or node-differential privacy (eg: \cite{kasiviswanathan2013analyzing}, \cite{blocki2013differentially}, etc. ) often considered in network privacy. In the Ising model, the network information incorporated into the $\bm J_n$ matrix is considered non-stochastic, and we are instead interested in the outcomes $\sigma_i\in \{\pm 1\}, \ i\in\{1,\cdots,n\}$ of the nodes. Taking up the tax-evasion example to elucidate, the choice to evade taxes, taken to be binary as $\pm 1$ (which can be affected by peers' choices), are sensitive and hence require privacy guarantees. This would give the respondents plausible deniability against financially criminal behavior, while still allowing the researchers to study the influence of peers in such behavioral models. 

To our knowledge, the work by \cite{zhang2020privately} is the only one discussing differential privacy in Ising models. They focus on keeping the dataset private during both structure learning and parameter estimation from multiple realizations of the results. 

However, their privacy concept differs from ours. 
They adopt a privacy model wherein the collection $\{\sigma_i\}_{i=1}^n$ treated as a singular unit of data. This approach is particularly designed for scenarios involving the observation of multiple independent replicates of datasets, each consisting of $n$ sign flips. In contrast, our approach considers each node's outcome as an individual unit, observing only a single collection of $n$ sign flips. This presents a more individualistic perspective on the preservation of privacy, making our applicability significantly different from theirs.

\textit{Our Contributions} can be summarized as follows:
\vspace*{-1em}
\begin{itemize}
    \setlength\itemsep{0em}
    \item Primarily we study the problem of preserving privacy for node outcomes of a network, in the context of parameter ($\beta$) estimation in an Ising model. This parameter is estimated with a single realization of the network, and can be used to infer about the extent of interference between node outcomes in a network. The problem of preserving node-outcome privacy in a single network data, as far as our knowledge is concerned, has not been studied before.
    \item We prove regret bounds for our algorithm, quantify the cost of privacy and complement the theoretical results with extensive simulation study with Erd\H{o}s-R\'enyi random graphs. 
    \item Finally, we evaluate the performance on two real-world networks-HIV status of individuals in a social network, and political leaning of online blogs that link to one another.
\end{itemize}
\vspace*{-0.5em}

The article is organised as follows: Section \ref{introduction} provides an introduction discussing the importance of privacy of node outcomes along-with the current state of the literature, Section \ref{formulation} puts the problem formally in terms of the model and the privacy guarantee being provided, Section \ref{ourmethod} discusses our algorithm and proves privacy and regret guarantees, and Section \ref{numericals} evaluates its performance through numerical experiments and real life data. Finally, Section \ref{discussion} provides closing discussions. All the proofs can be found in the Supplementary Material. 

\section{Problem Formulation} \label{formulation}

Two vectors $\bm \tau,\bm\tau' \in \{\pm 1\}^n$ are said to be adjacent if they differ in at most one coordinate. The notion of differential privacy tries to constraint an algorithm by limiting its output variability for adjacent training input $\bm\tau$ and $\bm\tau'$ (\cite{dwork2014algorithmic}). 

\begin{definition} (\cite{dwork2006differential, dwork2006our}) \label{pridef}
    A randomized algorithm $\mathcal{M}$ is said to be node outcome $(\vep,\delta)$-differentially private ($\vep>0$, $\delta\ge 0$) if
    $$\P(\mathcal M(\bm \sigma)\in S)\le e^{\vep}\P(\mathcal M(\bm \sigma') \in S) + \delta$$ for any adjacent vectors $\bm\sigma,\bm\sigma'\in \{\pm 1\}^n$ and all events $S$ in the output space of $\mathcal{M}$. When $\delta = 0$, the algorithm is said to be $\vep$-differentially private.
\end{definition}

Note that in Definition \ref{pridef}, we have not specified anything about the graph information. Indeed, the privacy protection is for the outcomes on the nodes, even when the graph information is perfectly available to an adversary.

Given a non-negative symmetric matrix $\bm J_n\in \R^{n\times n}$ (encapsulating network information) with 0 on its diagonal, the Ising model constitutes assigning a probability distribution on a vector of dependent $\pm 1$ random variables $\bm\sigma = (\sigma_1,\cdots,\sigma_n)$, given by a parametric distributions on $S_n:= \{-1,1\}^n$ given by
\vspace*{0em}
\begin{align} \label{pmf} \P_\beta(\bm\sigma = \bm\tau) = \dfrac 1{2^n}\exp\left(\dfrac 12\beta H_n(\bm\tau) - F_n(\beta)\right);  \end{align} with $\beta\geq 0$, where 
\begin{align} \label{hamiltonian}
    H_n(\bm\tau) = \bm\tau^T\bm J_n\bm \tau = \sum_{1\leq i, j\leq n} \bm J_n(i,j) \tau_i\tau_j; \ \ \bm\tau\in S_n
\end{align}
and $F_n(\beta)$ is the log-partition function determined by the normalizing constraint $\sum_{\bm\tau \in S_n} \P_\beta(\bm\sigma = \bm\tau) = 1$ resulting in the formulation 
\begin{align*}
    F_n(\beta) &:= \log\left[\dfrac{1}{2^n}\sum_{\tau\in S_n} \exp\left(\dfrac 12\beta H_n(\bm\tau)\right)\right] \\ &= \log\E_0\exp\left(\dfrac 12\beta H_n(\bm\sigma)\right)
\end{align*}
where $\E_0$ denotes the expectation over $\bm\sigma$ distributed as $\P_0$ (the uniform measure on $S_n$). 
The parameter $\beta$ , in parallel with the physics literature, is often known as the inverse temperature and captures the strength of dependence in the various entries of $\bm\sigma$. 

A very popular way of estimating $\beta$ is obtaining the maximum pseudo-likelihood estimator (MPLE) $\hat\beta_n(\bm\sigma)$ (\cite{bhattacharya2018inference, chatterjee2007estimation}), given by 
\begin{align}
    \hat\beta_n(\bm\sigma) = \arg\max_\beta \prod_{i=1}^n f_i(\beta,\sigma_i) 
\end{align}
where $f_i(\beta,\sigma_i)$ is the conditional probability density of $\sigma_i$ given $\bm\sigma_{-i}$, under parameter $\beta$.

For any $\bm\tau\in S_n$, defining the function  $L_{\bm\tau}:[0,\infty)\to \R$ as 
\begin{align} 
L_{\bm\tau}(x): = -\dfrac 1n\sum_{i=1}^n m_i(\bm\tau)(\tau_i - \tanh(xm_i(\bm\tau)), \label{def-lsigma}
\end{align} where 
\begin{align}\label{eq:defmitau}
     m_i(\bm\tau):= \sum_{j=1}^n \bm J_n(i,j) \tau_j, 
\end{align}
it can be verified (see for example, \cite{chatterjee2007estimation,bhattacharya2018inference}) that \begin{align} \hat\beta_n(\bm\sigma):=\inf\{x\geq 0: L_{\bm\sigma} (x) = 0\}, \label{MPLE} \end{align} interpreting the infimum of an empty set as $\infty$ as usual, where $\bm\sigma\sim\P_\beta$. Henceforth the dependence on $\sigma$ is suppressed with $\hat\beta_n := \hat\beta_n(\bm\sigma)$ denoting the MPLE of $\beta$, and the function defined in Equation (\ref{def-lsigma}) is referred to as the pseudo log-likelihood function. 
Furthermore, in the following we use the notation $t_{n} = \Theta(s_{n})$ to denote $t_{n} = O(s_{n})$ and $s_{n} = O(t_{n})$. 
Also we say a random variable $X_{n} = O_{p}(t_{n})$ to imply that for any $\vep>0$ there exists $M_{\vep}>0$ such that,
\begin{align*}
    \P\left[|X_{n}/t_{n}|>M_{\vep}\right]\leq \vep\text{ for all large enough }n.
\end{align*}

\section{Our Method} \label{ourmethod}

Our algorithm for private parameter estimation in one-parameter Ising model is given in Algorithm \ref{PrIsing}.

\begin{algorithm}[!ht]
    \caption{Private Estimation in One-parameter Ising Model (\texttt{PrIsing})}
    \begin{algorithmic}
        \Require $\bm\sigma = (\sigma_1,\cdots,\sigma_n)$, privacy parameters $\varepsilon>0, \delta\geq 0$,  symmetric matrix $\bm J_n\in \mathbb{R}^{n\times n}$ with non-negative entries such that $\bm J_n(i,i) = 0 \ \forall 1\leq i\leq n$.\\
        \State Set $m_i(\bm\sigma) = \sum_{j=1}^n \bm J_n(i,j)\sigma_j; \ i=1,\cdots, n$; $L_{\bm\sigma}(\beta) = -\frac 1n\sum_{i=1}^n m_i(\bm\sigma)(\sigma_i- \tanh(\beta m_i(\bm\sigma))$.\\
        \State Set $d_i = n\sum_{j=1}^n \bm J_n(i,j)$ for all $i=1(1)n$.\\
        \State Set $\zeta = \max_{j}\left\{8\dfrac{d_j}{n}\right\}$. \\
        \If{ $\delta>0$}
        \State Sample $b\in\R$ from $\nu(b;\varepsilon,\delta)=\mathcal{N}(0, \gamma^2)$ where $$\gamma = \dfrac{\zeta\sqrt{8\log (2/\delta) + 4\varepsilon}}{\varepsilon}$$\\
        \ElsIf{ $\delta = 0$}\\
        \State Sample $b\in\R$ from $\nu(b;\vep, 0) = \mathrm{Lap}\left(0,2\zeta/\vep\right)$ \\
        \EndIf\\
        \State Set $\Delta \geq \max_{j}\left\{\frac{24}{\varepsilon n}\sum_{i=1}^{n}d_{i}\bm J_{n}(i,j)\right\}$\\
        \State \Return $\hat\beta^{\p} = \inf\{\beta\geq 0: L_{\bm\sigma}(\beta) + \Delta\beta/n + b/n = 0\}$
    \end{algorithmic}
    \label{PrIsing}
\end{algorithm}

Although the non-private estimate is given by the MPLE obtained through equation (\ref{MPLE}), our algorithm builds on the MPLE method by equating the pseudo-likelihood equation not to 0, but to a random noise perturbation, calibrated according to the privacy requirement. The algorithm builds on \cite{kifer2012private} and uses similar proof ideas by bounding the ratio of the gradients of the MPLE equation, and the density of the noise. However, in the former ratio, they could use an identical bound as their data points were i.i.d., whereas due to the dependent structure of the MPLE equation ($m_i(\bm\sigma)$ depends on $\bm\sigma_{-i}$'s), we need to obtain the noise variance calibrated to the global-sensitivity (\cite{dwork2006calibrating}) of the pseudo-loglikelihood function, demonstrated in proof of Theorem \ref{PrIsingthm} in Section \ref{sec:proofofPrIsing}.
\begin{theorem}
    Given any $\epsilon>0$ and $\delta\ge 0$, Algorithm \ref{PrIsing} is $(\varepsilon,\delta)$-differentially private on node-outcome $\bm\sigma$.
    \label{PrIsingthm}
\end{theorem}

Next, we quantify the regret bound of our Algorithm \ref{PrIsing}, and quantify the cost of privacy in contrast with the non-private regret bound. Under regularity conditions, \cite{bhattacharya2018inference} shows $\sqrt{a_n}$ consistency of the non-private estimators, where $a_n$ is determined by conditions on the log-partition function. In the following result we adopt a conditions similar to those required for consistency of the non-private estimator and provide the regret bounds attained by $\hat{\beta}^{\p}$ from Algorithm \ref{PrIsing}.

\begin{theorem}[Simpler Version of Theorem \ref{thm:upperbdddetailed}] \label{upperbound}
    Let $\sup_{n\geq 1}\|\bm J_{n}\|<\infty$, and let $\beta_{0}>0$ be fixed. Suppose $\{a_{n}:n\geq 1\}$ is a sequence such that, $a_{n}\rightarrow\infty$ as $n\rightarrow\infty$ and,
    \begin{enumerate}
        \item [(i)] $F_{n}(\beta) = \Theta(a_{n})$ for all $\beta$ in a neighbourhood of $\beta_{0}$,
        \item [(ii)] $\E_{\beta_{0}}\left[\sum_{i=1}^{n}m_{i}(\bsigma)^2\right] = o(a_{n})$,
        \item [(iii)] $\sum_{i=1}^{n}\sum_{j=1}^{n}\bm J_{n}(i,j)^2 = O(a_{n})$.
    \end{enumerate}
    Then, whenever $\Delta$ is chosen to be the smallest permitted by Algorithm \ref{PrIsing}, the estimator $\hat{\beta}^{\p}$ satisfies,
    \begin{align*}
        |\widehat{\beta}^\p - \beta_{0}| &= O_{p}\left({\frac{1}{\sqrt{a_{n}}} + \frac{\lambda_{n}\sqrt{\log(2/\delta)}}{a_{n}\varepsilon}}\right) \text{ if }\delta> 0,
    \end{align*}
    and,
    \begin{align*}
        |\widehat{\beta}^\p - \beta_{0}| &= O_{p}\left({\frac{1}{\sqrt{a_{n}}} + \frac{\lambda_{n}}{a_{n}\varepsilon}}\right) \text{ if }\delta = 0.
    \end{align*}
    where $\lambda_n:= 1\vee \left\|\bm{J}_{n}\right\|_{1\rightarrow\infty}^2,$ with 
    $\|\cdot \|_{1\to\infty}$ denoting the vector induced 1-norm of a matrix.
\end{theorem}

The first term in the rate, denoted as $1/\sqrt{a_n}$, represents the non-private rate inherent in the problem. The additional term $\frac{\lambda_n}{a_n \varepsilon}$ accounts for the privacy cost introduced by differential privacy (DP) constraints. This is in line with the privacy literature, where it is widely observed that the privacy cost manifests as a higher-order term, and its dependence on $n$ (in our case, $a_n$) is quadratic in nature (as discussed in, for example, \cite{acharya2021differentially}).

However, It is important to note that \cite{acharya2021differentially} deal with problems exhibiting an independent and identically distributed (i.i.d.) structure. Therefore, while drawing analogies, one should consider this context carefully. We anticipate that in the presence of dependence structures among observations, privacy is compromised to a greater extent (\cite{liu2016dependence}). This compromise is quantified by the parameter $\lambda_n$, where larger values of $\lambda_n$ signify an exacerbation of the privacy cost.

\subsection{Examples}  

In this section we consider specific examples of the underlying network to quantify the regret bound obtained by the private estimator $\hat\beta^{\p}$ and contrast with the corresponding non-private estimator.

\subsubsection{Degree Regular Graphs}
For Ising Models on degree regular graphs $G_{n}$, the matrix $\bmJ_{n}$ in \eqref{pmf} becomes $\bmJ_{n} = \bm{A}_{n}/D_{n}$ where $\bm{A}_{n}$ is the adjacency matrix of the graph $G_{n}$ and $D_{n}$ is the degree. This encompasses Ising models on complete graphs, random regular graphs and lattices which have been comprehensively investigated in probability and statistical physics (see \cite{dembo2010ising,levin2010glauber}). For such $\bmJ_{n}$ the parameter $\lambda_{n}= 1$, and hence by Theorem \ref{upperbound} and Corollary 3.1 from \cite{bhattacharya2018inference} we have the following result.

\begin{corollary}\label{cor:degree-regular-rates}
    Fix $\beta_{0}>0$. Then for any sequence of $D_{n}$ regular graphs $G_{n}$,
    \begin{align*}
        \left|\hat{\beta}^{\p} - \beta_{0}\right| = 
        \begin{cases}
            O_{p}\left(\sqrt{\frac{D_{n}}{n}} + \frac{D_{n}}{n\vep}\eta_{\delta}\right) & 0<\beta_{0}<1\\
            O_{p}\left(\frac{1}{\sqrt{n}} + \frac{1}{n\vep}\eta_{\delta}\right) & \beta_{0}>1
        \end{cases}
    \end{align*}
    where $\eta_{\delta} = \log(2/\delta)$ if $\delta>0$, and $\eta_{\delta} = 1$ if $\delta = 0$.
\end{corollary}

\subsubsection{Erd\H{o}s-R\'enyi Graphs}
Consider $G_{n}$ to be a sequence of Erd\H{o}s-R\'enyi random graphs on $n$ vertices with edge probability $p_{n}$. For Ising Models with such an underlying network structure the matrix $\bmJ_{n}$ from \eqref{pmf} is taken as $\bm A_{n}/np_{n}$, where $\bm A_{n}$ is the adjacency matrix of the network $G_{n}$. It is easy to infer that for large enough $n$ the parameter $\lambda_{n}\leq 2$ with high probability. Combining Theorem \ref{upperbound} and Corollary 3.2 from \cite{bhattacharya2018inference} we have the following result.
\begin{corollary}\label{cor:erdos-rates}
    Fix $\beta_{0}>0$. Then for a sequence of Erd\H{o}s-R\'enyi random graphs $G_{n}$ with edge probability $\frac{\log n}{n}\ll p_{n}\leq 1$,
    \begin{align*}
        \left|\hat{\beta}^{\p} - \beta_{0}\right| = 
        \begin{cases}
            O_{p}\left(\sqrt{p_{n}} + \frac{p_{n}}{\vep}\eta_{\delta}\right) & 0<\beta_{0}<1\\
            O_{p}\left(\frac{1}{\sqrt{n}} + \frac{1}{n\vep}\eta_{\delta}\right) & \beta_{0}>1
        \end{cases}
    \end{align*}
    where $\eta_{\delta} = \log(2/\delta)$ if $\delta>0$, and $\eta_{\delta} = 1$ if $\delta = 0$.
\end{corollary}

In Corollary \ref{cor:degree-regular-rates} and \ref{cor:erdos-rates}, we observe similar phase transition behavior as in non-private scenarios, with a distinct change in the rate of convergence occurring at $\beta_0 = 1$. Specifically, when $\beta_0 < 1$, we witness a phenomenon reminiscent of mean estimation problems, where the observed rate follows $O_p(\sqrt{\frac{d}{n}} + \frac{d}{n\varepsilon})$ (as discussed in \cite{cai2021cost}), where $d$ represents the dimensionality of the problem. In these regimes, $D_n$ and $np_{n}$ signifies the intrinsic dimensionality of our problem. Notably, the cost of privacy becomes more pronounced for larger values of $D_n$ and $np_{n}$, indicating a higher degree of dependence in our model. However, intriguingly, this intensification of the privacy cost diminishes in the high-dependence regime $\beta_0 > 1$. Here, the cost of privacy no longer exhibits dependency on the graph's intrinsic dimensionality, paralleling the non-private rate.

\section{Numerical Experiments} \label{numericals}

To complement the regret guarantees in Section \ref{ourmethod}, we perform numerical experiments to evaluate the performance of our method. We conduct a set of simulations on Erd\H{o}s-R\'enyi graphs and on two real datasets- HIV status of individuals in a social network, and political leaning of online blogs, repeating the simulation 500 times to plot the results. The Erd\H{o}s-R\'enyi simulations show a regime change in the estimation rate of $\beta$, a phenomenon also seen in the nonprivate setting \cite{bhattacharya2018inference}. On the other hand, the real network experiments show that the validity of the method is upheld in realistic networks as well.

\subsection{Experiments on Erd\H{o}s-R\'enyi graphs}

In this section, we rigorously validate our theory through a comprehensive simulation study. We begin by providing a detailed description of the simulation setup.

\subsubsection*{Simulation Setup}

We aim to estimate the parameter $\beta$ based on a dataset consisting of $n$ observations. These observations are generated from an Ising model, where the underlying dependency graph $G_n$ is created using the Erdos-Renyi model with a designated parameter $p_n$. $\bm J_n$ is taken to be $\bm A(G_n)/np_n$, where $\bm A(G)$ is taken to be the adjacency matrix of a network $G$, ie, the corresponding $H$ from \eqref{hamiltonian} becomes
\vspace*{-0.5em}
$$H(\bm\tau) = \dfrac{1}{np_n}\bm\tau^T \bm A(G_n)\bm\tau.\vspace{-0.5em}$$
Our investigation into the performance of our proposed estimator encompasses three primary simulation studies, as discussed below.

\paragraph{1. Impact of $\beta$ on Our Estimator}
We want to compare the performance of our $\hat\beta^\p$ with $\hat\beta_n$ over a range of true $\beta$. We set $n = 2000$, $\delta = 1/n$ and $\vep = 5$ for the comparison, and $p_n$ is chosen to be $n^{-\frac{1}{3}}$. 

\begin{figure}[!ht]
    \centering
    \includegraphics[scale = 0.7]{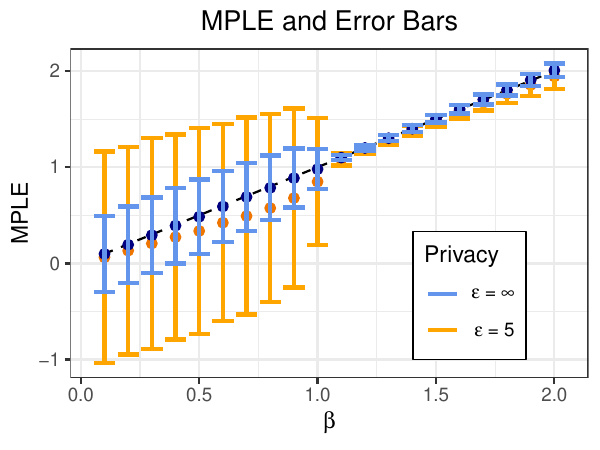}
    \caption{Private and non-private MPLE in an Ising model on and Erd\H{o}s-R\'enyi random graph.}
    \label{beta-erg}
\end{figure}
Figure \ref{beta-erg} shows the performance of the estimator over $\beta\in [0,2]$ alongwith 1 standard deviation errorbars. Notice the phase transition at 1 in the error-bars produced in both the non-private and private estimators. This is in line with what we expect in theory, as $\hat\beta_n$ is $\frac{1}{\sqrt{p_{n}}} = n^{\frac{1}{6}}$ consistent for $\beta<1$ and $\sqrt n$ consistent for $\beta>1$ (see Corollary \ref{cor:erdos-rates}), and the private estimator follows the same trend.

\paragraph{2. Effect of the Number of Observations $n$ on Mean Squared Error (MSE)}
Next we focus on how the MSE of the estimators vary with the number of nodes $n$ in $G_n$. $p_n$ $\delta$ are still taken to be as before, and we use a range of $\vep$ for comparison. Since the rate of consistency is different in the two regimes of $\beta$, we plot the MSE vs $n$ at $\beta =0.5$ (Figure \ref{mse0.5-erg}) and at $\beta =1.5$ (Figure \ref{mse1.5-erg}). 
\begin{figure}[!ht]
    \centering
    \includegraphics[scale = 0.7]{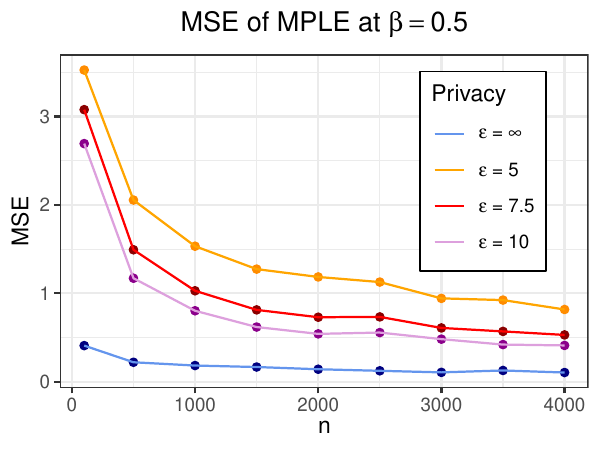}
    \caption{Effect of $n$ on MSE of MPLE in an Ising model on and Erd\H{o}s-R\'enyi random graph with $\beta=0.5$}
    \label{mse0.5-erg}
\end{figure}

\begin{figure}[!ht]
    \centering
    \includegraphics[scale = 0.7]{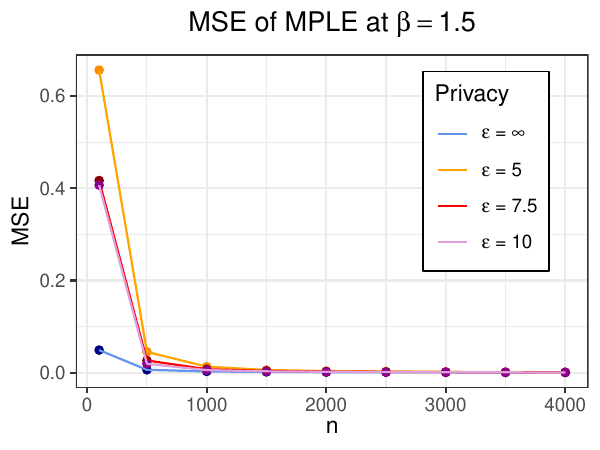}
    \caption{Effect of $n$ on MSE of MPLE in an Ising model on and Erd\H{o}s-R\'enyi random graph with $\beta=1.5$}
    \label{mse1.5-erg}
\end{figure}

Note that both Figure \ref{mse0.5-erg} and \ref{mse1.5-erg} show a downward trend in MSE as expected, but the speed at which the trend dips down varies over $\vep$. This effect can be attributed to the cost of privacy in Corollary \ref{cor:erdos-rates}, in particular for $\beta<1$ the cost is $\frac{p_{n}}{\vep} = \frac{1}{n^{1/3}\vep}$ while for $\beta>1$ the cost is $\frac{1}{n\vep}$, which complements the observation that for $\beta>1$ the reduction in MSE is much faster than $\beta<1$.

\paragraph{3. Effect of Edge Density of $G_n$ on MSE}
Here, we investigate the relationship between the edge density of the underlying graph $G_n$ and the resulting Mean Squared Error (MSE) of our estimator. Since the number of edges is expected to be around $n^2p_n$, we take $p_n=n^{-\alpha}$, and vary $\alpha$ to contrast the edge densities.

\begin{figure}[!ht]
    \centering
    \includegraphics[scale = 0.7]{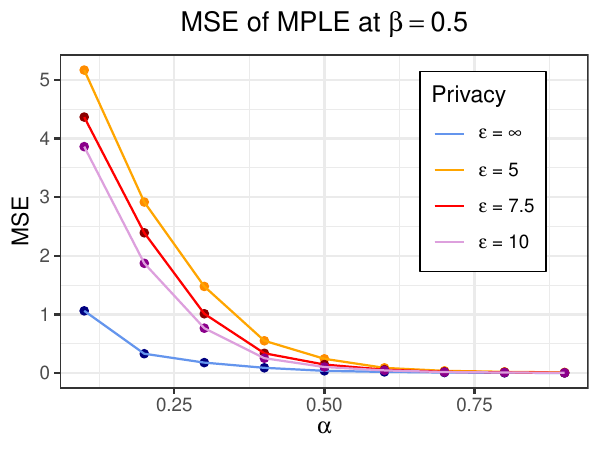}
    \caption{Effect of $p_n$ on MSE of MPLE in an Ising model on and Erd\H{o}s-R\'enyi random graph with $\beta=0.5$}
    \label{mse0.5dens-erg}
\end{figure}

\begin{figure}[!ht]
    \centering
    \includegraphics[scale = 0.7]{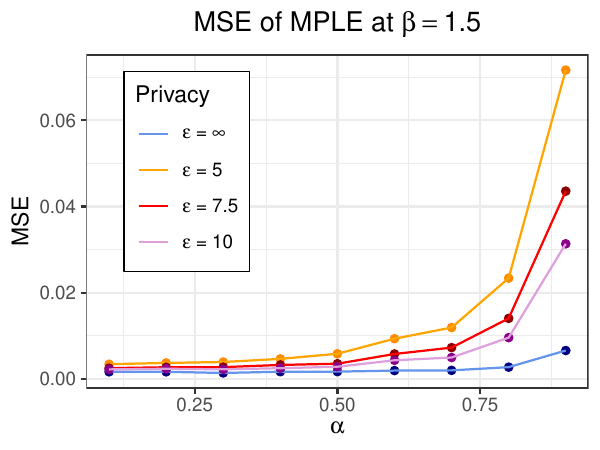}
    \caption{Effect of $p_n$ on MSE of MPLE in an Ising model on and Erd\H{o}s-R\'enyi random graph with $\beta=1.5$}
    \label{mse1.5dens-erg}
\end{figure}

Figure \ref{mse0.5dens-erg} and \ref{mse1.5dens-erg} show the MSE of MPLE for a range of $\alpha$. Recall from Corollary \ref{cor:erdos-rates} that rate of convergence in the high temperature regime $(\beta<1)$ is inversely proportional to $\alpha$, which explains the relation betweeen MSE and $\alpha$ in Figure \ref{mse0.5dens-erg}. On the other hand in the low temperature regime $(\beta>1)$ the rate of convergence is $\sqrt{n}$, independent of the choice of $\alpha$, reflected in non-private curve in Figure \ref{mse1.5dens-erg}. However, in the private case, following Theorem \ref{upperbound},  the increment in error with an increasing $\alpha$ can be attributed to the parameter $\lambda_{n}$, which is approximately $1$ with additional error proportional to $\alpha$ with high probability.
\subsection{Real world networks: Experiments \& Real data}

In the second set of simulations, we adopt real networks from two datasets, HIV transmission in social networks, and political affiliations of online blogs. We conduct synthetic experiments adopting the corresponding networks as fixed, and perform simulations of Ising model realizations on these networks with $\bm J_n = \bm D(G_n)^{-1/2} \bm A{(G_n)}\bm D(G_n)^{-1/2}$, where $\bm D(G) = \text{diag}(D_{1}(G),\cdots, D_{n}(G))$ is a diagonal matrix of the degrees of the nodes in a graph $G$. This leads to 
\begin{align} H(\bm\tau) = \bm\tau^T \bm D(G_n)^{-\frac 12}\bm A(G_n)\bm D(G_n)^{-\frac 12}\bm\tau \label{lapla-scaling} \end{align} 
which can thus handle moderate degree heterogeneity in the network $G_n$, and is a generalization of the scaling used for regular or Erd\H{o}s-R\'enyi graphs. In fact this choice of $\bm J_n$ can be linked to the normalized graph Laplacian $\mathcal{\bm L}_n$ as $\bm J_n = \bm I_n - \mathcal{\bm L}_n$ (\cite{chung1997spectral}), and have been extensively used in node label classification problems (\cite{li2018deeper}, \cite{zhou2020graph, wu2020comprehensive}).

\subsubsection{HIV status of individuals in a social network} \label{HIVsubsect}

We consider the network of HIV status of individuals in Colorado Springs with $403$ individuals in the years 1988-1993 pooled together (\cite{morris2011hiv}), of which $23$ have HIV status positive. Clearly this is a network where the privacy of the node outcomes is of great importance, and the outcomes are heavily imbalanced in the network, as given by the numbers as well as the network plot in Figure \ref{HIVnetwork}. 

\begin{figure}[!ht]
    \centering
    \includegraphics[scale = 0.5]{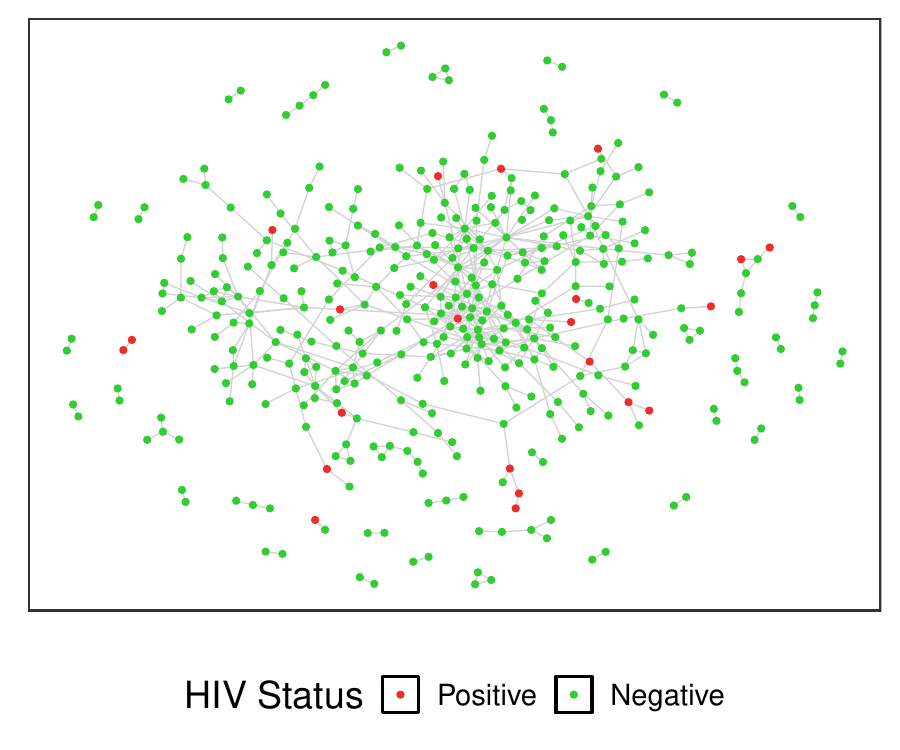}
    \caption{Social Network with HIV Status}
    \label{HIVnetwork}
\end{figure}

We conduct synthetic experiments simulating Ising model realizations from this network under the Laplacian scaling as in Equation \ref{lapla-scaling}, and plot the results in Figure \ref{HIV:data-sims}. $\vep = 5$ and $\delta = 1/n$ are chosen for the plot. Both the private as well as non-private estimate appears to be consistent around the true $\beta$.

\begin{figure}[!ht]
    \centering
    \includegraphics[scale = 0.7]{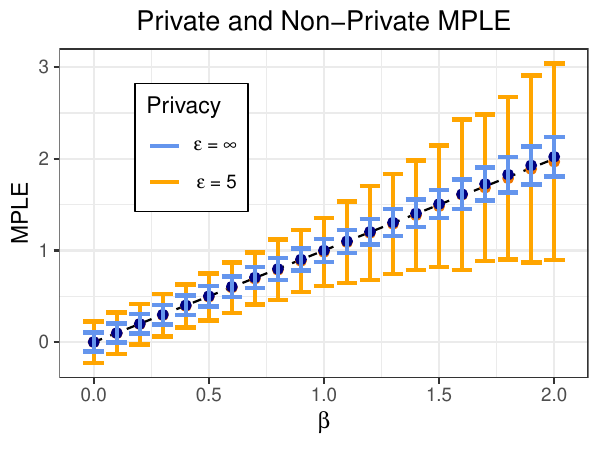}
    \caption{Performance of private and non-private estimators based on Ising model synthetic data on real HIV status network.}
    \label{HIV:data-sims}
\end{figure}

Next we perform the analysis of the real data. The non-private ${\hat\beta_n = 1.8}$, and we produce $\hat\beta^\p$ and take the Monte-Carlo conditional expectation of $\E[(\hat\beta^\p-\hat\beta_n)^2|\bm\sigma,\bm J_n]$ to quantify the cost of privacy. 
It can be seen from Figure \ref{HIV:MSEeps-real} that the cost of privacy has a decreasing trend over $\vep$, which is as expected. 


\begin{figure}[!ht]
    \centering
    \includegraphics[scale = 0.7]{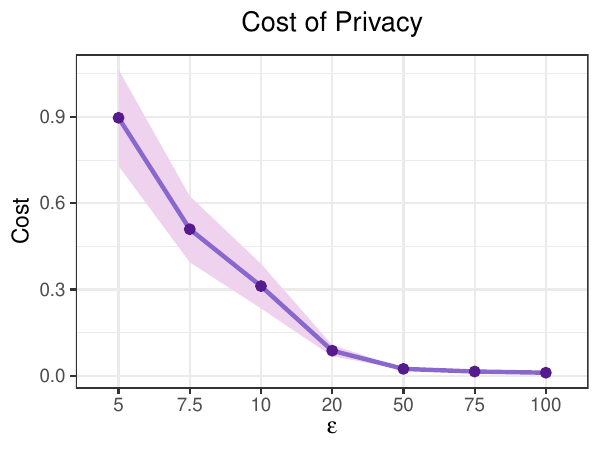}
    \caption{Cost of privacy across $\vep$ on the HIV status network and real data}
    \label{HIV:MSEeps-real}
\end{figure}

\subsubsection{Political leaning of online blogs}\label{sec:Polblogs}

Next we consider the network of popular political blogs over the period of two months preceding the U.S. Presidential Election of 2004 (\cite{adamic2005political}) and their political leaning. The graph, plotted in \ref{polblogsnetwork}, have nodes representing the blogs, color coded by their political leaning, and edges between two nodes if and only if at least one of them link to the other.  We have removed nodes with very high degrees $(\ge 50)$ as they are very popular blogs anyway (like \texttt{blogforamerica.com, churchofcriticalthinking.com, brilliantatbreakfast.blogspot.com, busybusybusy.com,}etc.) and are outliers in measuring the influence of the linking network on the political leaning. Any isolated node have also been removed to create a connected graph. The removed nodes are relatively balanced on both sides of the political spectrum. The resulting network,  of $n= 815$ nodes is relatively balanced in the two outcomes, liberal (382) and conservative (433), and we want to maintain the privacy of the political leanings of the blogs while measuring the influence of the links on a blog being red or blue. 


\begin{figure}[!ht]
    \centering
    \includegraphics[scale = 0.5]{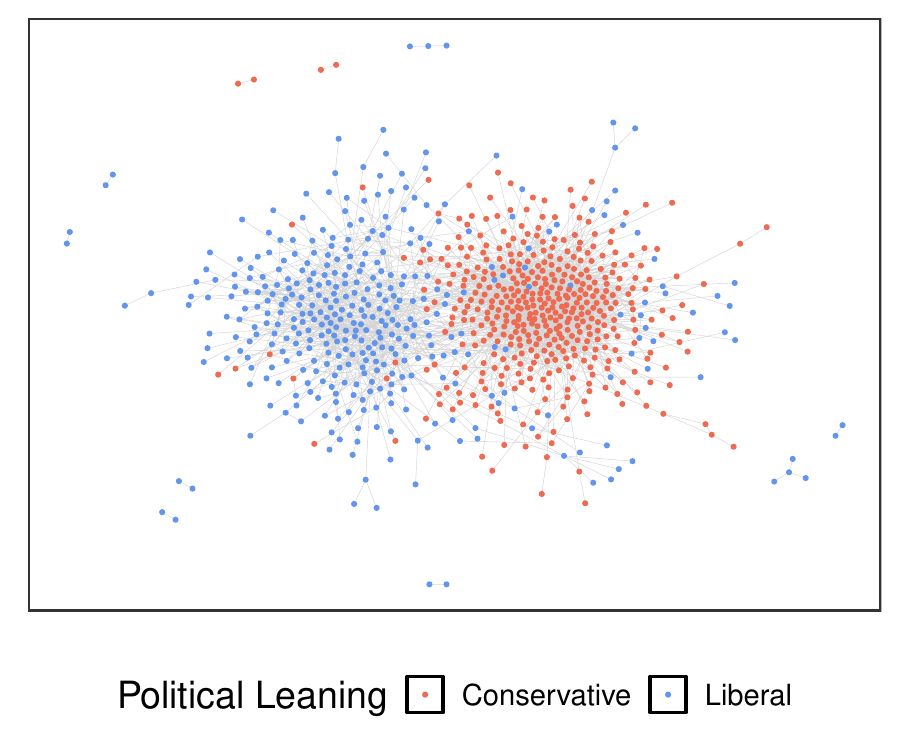}
    \caption{Link Network between Political Blogs}
    \label{polblogsnetwork}
\end{figure}

As before we conduct synthetic experiments simulating Ising model on this network, and the results in Figure \ref{Polblogs:data-sims} show how the estimates, both private and non-private concentrate around the true $\beta$. 

\begin{figure}[!ht]
    \centering
    \includegraphics[scale = 0.7]{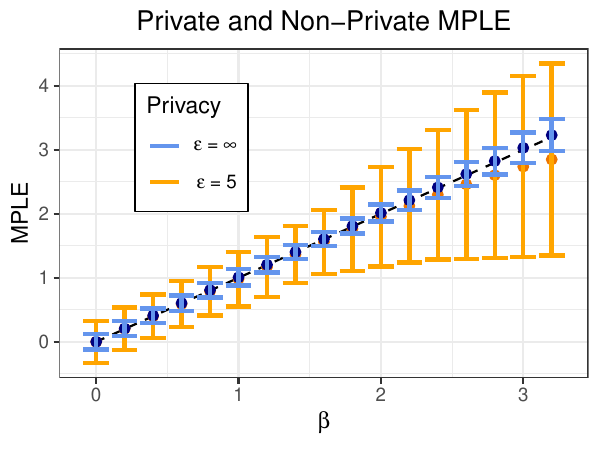}
    \caption{Performance of private and non-private estimators based on Ising model synthetic data on real political blogs network.}
    \label{Polblogs:data-sims}
\end{figure}

In the real data $\hat\beta_n = 2.85$ here, and as before we conduct the cost of privacy analysis as in Section \ref{HIVsubsect}. The results plotted in Figure \ref{Polblogs:MSEeps-data} show the expected downward trend of MSE for rising $\vep$.

\begin{figure}[!ht]
    \centering
    \includegraphics[scale = 0.7]{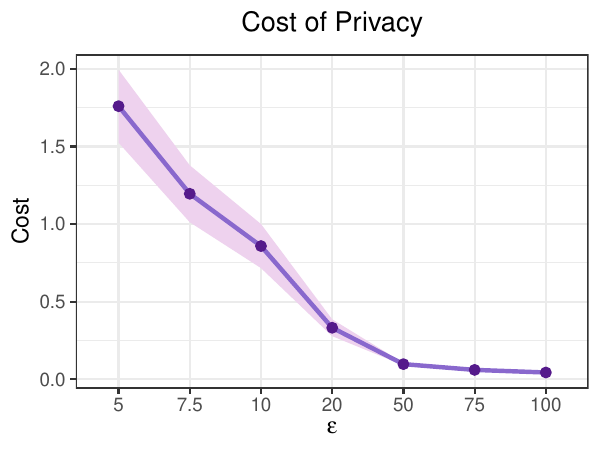}
    \caption{Cost of privacy across $\vep$ on the political blogs network and real data.}
    \label{Polblogs:MSEeps-data}
\end{figure}

\section{Discussion} \label{discussion}

The Ising model, initially developed for ferromagnetism, has wide applications in various fields, including social sciences and healthcare, for modeling outcome dependencies in networked systems. However, its use in contexts like disease transmission, tax evasion, and social behavior raises privacy concerns.

Current privacy research primarily focuses on independent models, whereas network analysis mainly employs edge and node differential privacy. However, the Ising model presents a unique challenge of protecting interdependent node outcomes alongside network structure, which current literature does not adequately address.

To address this gap, we introduced an $(\varepsilon,\delta)$-differentially private algorithm for Ising models, ensuring node outcome privacy with a single network realization. Our work contributes theoretical insights, including the consistency of the maximum pseudo-likelihood estimator and quantifying privacy cost as $O\left(\frac{\lambda_n}{a_n \varepsilon}\right)$.

Our experiments demonstrate the algorithm's practicality, preserving privacy and utility in synthetic and real-world networks, like HIV status in a social network and political leaning of online blogs.

Despite our contributions, limitations remain, particularly assumptions related to log-partition functions (see \cite{dagan2021learning}). Possible avenues for exploration may involve integrating network privacy and node outcome protection while evaluating privacy guarantees for different privacy notions, such as Renyi Differential Privacy or $f$-DP (see \cite{dong2022gaussian}).

In summary, our research emphasizes the importance of privacy in Ising models. Our $(\varepsilon,\delta)$-differentially private algorithm addresses this concern effectively, offering valuable contributions to this critical field. We hope this work encourages further exploration of privacy preservation techniques in Ising models, promoting a more secure approach to network analysis.

\bibliography{ref.bib}

\onecolumn

\aistatstitle{PrIsing: Supplementary Materials}

\appendix
\section{Proof of Theorem \ref{PrIsingthm}}\label{sec:proofofPrIsing}
In this section we prove that Algorithm \ref{PrIsing} satisfies $(\vep,\delta)$ privacy for any $\vep>0$ and $\delta\geq 0$. The proof is organised as follows. First, in the following lemma, we bound the amount of change in $L_{\bsigma}(\beta)$ induced by flipping a coordinate in $\bsigma$. 
\begin{lemma}\label{mdiff-bound}
    Fix $1\leq j\leq n$. If $\bsigma$ and $\bsigma'$ differ in only the $j$-th entry, then for any $\beta>0$, 
    \begin{align*}
        |L_{\bsigma}(\beta) - L_{\bsigma'}(\beta)|\le 8\dfrac{d_j}{n^2}.
    \end{align*}
\end{lemma}
The proof of Lemma \ref{mdiff-bound} is provided in Section \ref{subsec:prooflemmadiffbdd}. Now consider $\bsigma,{\bsigma}'\in \{-1,1\}^{n}$ such that,
\begin{align}\label{eq:defsigmaprime}
    \sum_{i=1}^{n}\one\left\{{\sigma}_{i}\neq {\sigma}_{i}'\right\}=1.
\end{align}
Fix $\alpha>0$. If $\hat\beta^{\p} = \alpha$, then $\alpha$ must satisfy,
\begin{align*}
    L_{\bsigma}(\alpha) + \Delta\alpha/n + b/n = 0
\end{align*}
or equivalently, define
\begin{align}\label{eq:defbalpha}
    b(\alpha;{\bsigma}) := -(nL_{\bsigma}(\alpha) +\Delta\alpha).
\end{align}
Using a change of variable approach the ratio of densities can be written as,
\begin{align}\label{eq:bdddenratio}
    \dfrac{f_{\widehat\beta^\p,{\bsigma}}(\alpha)}{f_{\widehat\beta^\p,{\bsigma}'}(\alpha)} = \dfrac{\nu(b(\alpha,{\bsigma});\varepsilon,\delta)}{\nu(b(\alpha,{\bsigma}');\varepsilon,\delta)}\cdot\dfrac{|\nabla b(\alpha;{\bsigma}')|}{|\nabla b(\alpha;{\bsigma})|},
\end{align} 
where $\nabla$ denotes the partial derivative with respect to $\alpha$ and $f_{\widehat{\beta}^{\p},\bm\tau}$ denotes the density of $\widehat{\beta}^{\p}$ given the data $\bm\tau = {\bsigma},{\bsigma}'$. Now in the subsequent lemmas we bound the two ratios appearing in R.H.S of \eqref{eq:bdddenratio} separately. First, in the following result, we bound the second term.
\begin{lemma}\label{lemma:bdjacobian}
For any ${\bsigma},{\bsigma}'\in \{-1,1\}^{n}$ satisfying \eqref{eq:defsigmaprime},
\begin{align*}
    \left|\dfrac{\nabla b(\alpha;{\bsigma})}{\nabla b(\alpha;{\bsigma}')}\right|\leq e^{\frac{\vep}{2}}
\end{align*}
where $b(\alpha,\cdot)$ is defined in \eqref{eq:defbalpha}.
\end{lemma}
Next, we provide a bound on the ratio of densities in the following lemma.
\begin{lemma}\label{lemma:bddensityratio}
Consider any ${\bsigma},{\bsigma}'\in \{-1,1\}^{n}$ satisfying \eqref{eq:defsigmaprime}. Then using Algorithm \ref{PrIsing} for $(\vep,\delta)$ privacy with $0<\delta<\frac{2}{\sqrt{e}}$, we get, 
\begin{align*}
    \dfrac{\nu(b(\alpha,{\bsigma});\varepsilon,\delta)}{\nu(b(\alpha,{\bsigma}');\varepsilon,\delta)}\leq e^{\vep/2}
\end{align*}
on a set $S\subseteq\R$ such that $\P\left(b(\alpha,{\bsigma})\in S\right)\geq 1-\delta$, and for $(\vep,0)$ privacy we get,
\begin{align*}
    \dfrac{\nu(b(\alpha,{\bsigma});\varepsilon,0)}{\nu(b(\alpha,{\bsigma}');\varepsilon,0)}\leq e^{\vep/2}.
\end{align*}
\end{lemma}
The proofs of Lemma \ref{lemma:bdjacobian} and Lemma \ref{lemma:bddensityratio} are given in Sections \ref{sec:prfjacobian} and \ref{sec:prfdenratio} respectively. We now proceed to show that Algorithm \ref{PrIsing} preserves the notion of $(\vep,\delta)$ differential privacy as defined in Definition \ref{pridef}. The proof of $(\vep,0)$ differential privacy is now immidiate by combining the bounds from Lemma \ref{lemma:bdjacobian}, Lemma \ref{lemma:bddensityratio} and \eqref{eq:bdddenratio}. For $\delta>0$, recalling $S$ from Lemma \ref{lemma:bddensityratio} observe that,
\begin{align*}
    f_{\widehat{\beta}^{\p},{\bsigma}}(\alpha)
    &\leq e^{\vep/2}f_{\widehat{\beta}^{\p},{\bsigma}'}(\alpha)\one\left\{b(\alpha,{\bsigma})\in S\right\} + f_{\widehat{\beta}^{\p},{\bsigma}}(\alpha)\one\left\{b(\alpha,{\bsigma})\in S^{c}\right\}\\
    &\leq e^{\vep/2}f_{\widehat{\beta}^{\p},{\bsigma}'}(\alpha) + f_{\widehat{\beta}^{\p},{\bsigma}}(\alpha)\one\left\{b(\alpha,{\bsigma})\in S^{c}\right\}
\end{align*}
Then for any borel set $A\subseteq\R$ and using a change of variable we get,
\begin{align*}
    \int_{A}f_{\widehat{\beta}^{\p},{\bsigma}}(\alpha)\mathrm{d}\alpha
    &\leq e^{\vep/2}\int_{A}f_{\widehat{\beta}^{\p},{\bsigma}'}(\alpha)\mathrm{d}\alpha + \int f_{\widehat{\beta}^{\p},{\bsigma}}(\alpha)\one\left\{b(\alpha,{\bsigma})\in S^{c}\right\}\mathrm{d}\alpha\\
    &\leq e^{\vep/2}\int_{A}f_{\widehat{\beta}^{\p},{\bsigma}'}(\alpha)\mathrm{d}\alpha + \int \nu\left(b(\alpha,{\bsigma});\vep,\delta\right)\one\left\{b(\alpha,{\bsigma})\in S^{c}\right\}\mathrm{d}b(\alpha,{\bsigma})\\
    &\leq e^{\vep/2}\int_{A}f_{\widehat{\beta}^{\p},{\bsigma}'}(\alpha)\mathrm{d}\alpha + \delta
\end{align*}
where the last bound follows by definiton of $S$ from Lemma \ref{lemma:bddensityratio}. The proof is now completed by recalling the choice of ${\bsigma}$ and ${\bsigma}'$ from \eqref{eq:defsigmaprime}.
\subsection{Proof of Lemma \ref{mdiff-bound}}\label{subsec:prooflemmadiffbdd}
        Recalling \eqref{eq:defmitau}, note that $m_i({\bsigma})$ does not depend on ${\sigma}_i$ for all $1\leq i\leq n$. Using \eqref{def-lsigma}, we have
            \begin{align*}
                L_{\bm\tau}(\beta) &= -\frac{1}{n}\sum_{i=1}^n m_i(\bm\tau)\tau_i + \frac{1}{n}\sum_{i=1}^n m_i(\bm\tau)\tanh(\beta m_i(\bm\tau))
            \end{align*}
            for $\bm\tau = {\bsigma},{\bsigma}'$.
            Then,
            \begin{align}
                L_{\bsigma}(\beta) - L_{{\bsigma}'}(\beta) &= -\frac{1}{n}\sum_{i: i\neq j} [m_i({\bsigma}) - m_i({\bsigma}')]{\sigma}_i - \frac{1}{n}m_j({\bsigma})({\sigma}_j -{\sigma}_j') \nonumber \\ & \hspace{1cm}- \frac{1}{n}\sum_{i=1}^n [m_i({\bsigma}^{\prime})\tanh(\beta m_i({\bsigma}^{\prime})) - m_i({\bsigma})\tanh(\beta m_i({\bsigma}))] \label{difl-1}
            \end{align}
        Now recalling that all entries of $\bmJ_{n}$ are non-negative, 
        \begin{align}\label{eq:boundonmisigmadiff}
            |m_i({\bsigma}) - m_i({\bsigma}')| = |\sum_{k=1}^n ({\sigma}_k - {\sigma}_k')\bmJ_n(i,k)| = |({\sigma}_j - {\sigma}_j')\bmJ_n(i,j)|\le 2\bmJ_n(i,j).
        \end{align}
        Consider,
        \begin{align*}
            \kappa(x) := x\tanh(\beta x),\ \forall x\in\mathbb{R}.
        \end{align*}
        It is easy to see that,
        \begin{align*}
            \kappa^{'}(x) = \tanh(\beta x) + x\beta\sech^{2}(\beta x),\ \forall x\in\mathbb{R},
        \end{align*}
        and by definition $|\kappa^{'}|\leq 2$. Then using the Mean Value Theorem we conclude,
        \begin{align*}
            |\kappa(x) - \kappa(y)|\leq 2|x-y|\text{ for all }x,y\in\mathbb{R}.
        \end{align*}
        Now recalling the definition of $\kappa$ and \eqref{eq:boundonmisigmadiff} we get,
        \begin{align}\label{eq:boundonmitanhdiff}
            |m_i({\bsigma}')\tanh(\beta m_i({\bsigma}')) 
            & - m_i({\bsigma})\tanh(\beta m_i({\bsigma}))| = |\kappa(m_{i}({\bsigma})) - \kappa(m_{i}({\bsigma}'))|\leq 4\bmJ_{n}(i,j)
        \end{align}
        Thus combining \eqref{difl-1},\eqref{eq:boundonmitanhdiff} and noticing that $|m_j({\bsigma})|\le \sum_{k=1}^n \bmJ_n(j,k) \le d_j/n$, we have,
        \begin{align*}
            |L_{\bsigma}(\beta) - L_{{\bsigma}'}(\beta)|&\le \frac{2}{n}\sum_{i=1}^n \bmJ_n(i,j) + \frac{2}{n}|m_j({\bsigma})| + \frac{4}{n}\sum_{i=1}^n \bmJ_n(i,j)\le 2\dfrac{d_j}{n^2} + 2\dfrac{d_j}{n^2} + 4\dfrac{d_j}{n^2}=8\dfrac{d_j}{n^2},
        \end{align*}
        completing the proof of the lemma.

\subsection{Proof of Lemma \ref{lemma:bdjacobian}}\label{sec:prfjacobian}
Since $\bsigma$ and $\bsigma'$ satisfy \eqref{eq:defsigmaprime}, then there exists $1\leq j\leq n$ such that $\sigma_{j}\neq \sigma_{j}'$. Note that,
    \begin{align}
        \left|\dfrac{\nabla b(\alpha,{\bsigma})}{\nabla b(\alpha,{\bsigma}')}\right|\leq 1+\left|\dfrac{\nabla b(\alpha,{\bsigma}) - \nabla b(\alpha,{\bsigma}')}{\nabla b(\alpha,{\bsigma}')}\right|.
        \label{bd1}
    \end{align}
    Once again by \eqref{eq:defmitau}, note that $m_i({\bsigma})$ does not depend on ${\bsigma}_i$ for all $1\leq i\leq n$. Now recalling \eqref{eq:defbalpha} and taking derivative on both sides of \eqref{difl-1} shows,
    \begin{align}
        |\nabla b(\alpha,{\bsigma}) - \nabla b(\alpha,{\bsigma}')|
        &= \left|\sum_{i:i\neq j} m_i({\bsigma})^2\sech^2\left(\alpha m_i({\bsigma})\right) - m_i({\bsigma}')^2\sech^2\left(\alpha m_i({\bsigma}')\right)\right|\nonumber\\
        & \leq \sum_{i:i\neq j}\left|m_i({\bsigma})^2\sech^2\left(\alpha m_i({\bsigma})\right) - m_i({\bsigma}')^2\sech^2\left(\alpha m_i({\bsigma}')\right)\right|\nonumber\\
        &\leq \frac{2}{n}\sum_{i:i\neq j}d_{i}\left|m_{i}({\bsigma})\sech(\alpha m_{i}({\bsigma})) - m_{i}({\bsigma}')\sech(\alpha m_{i}({\bsigma}'))\right|
        \label{delbdiffbound}
    \end{align}
    where the inequality in \eqref{delbdiffbound} follows from the bounds $|m_{i}({\bsigma})|\leq d_{i}/n$ and $|\sech(\cdot)|\leq 1$. Define,
    \begin{align*}
        \kappa_{0}(x):= x\sech(\alpha x)\ \forall x\in\mathbb{R}.
    \end{align*}
    Observe that,
    \begin{align*}
        \kappa_{0}'(x) = \sech(\alpha x) - x\alpha\sech(\alpha x)\tanh(\alpha x)\ \forall x\in\mathbb{R}.
    \end{align*}
    Now it is easy to infer that $|\kappa_{0}'(\cdot)|\leq 3$. Using Mean value theorem we get,
    \begin{align}\label{eq:MVTkappa0}
        \left|\kappa_{0}(x) - \kappa_{0}(y)\right|\leq 3|x-y|\ \forall x,y\in\mathbb{R}.
    \end{align}
    Finally by the definition of $\kappa_{0}$, \eqref{eq:MVTkappa0}, \eqref{eq:boundonmisigmadiff} and recalling that entries of $\bmJ_{n}$ are non-negative we conclude,
    \begin{align}\label{eq:misechbdd}
        \left|m_{i}({\bsigma})\sech(\alpha m_{i}({\bsigma})) - m_{i}({\bsigma}')\sech(\alpha m_{i}({\bsigma}'))\right|\leq 6\bmJ_{n}(i,j), \ \forall i\neq j.
    \end{align}
    Next, note that 
    \begin{align*}
        |\nabla b(\alpha,{\bsigma}')| = |\Delta + \sum_{i=1}^n m_i({\bsigma}')^2\sech^2(\alpha m_i({\bsigma}'))|\geq \Delta
    \end{align*}
    Thus recalling \eqref{bd1}, \eqref{delbdiffbound} and \eqref{eq:misechbdd} shows,
    \begin{align*}
        \left|\dfrac{\nabla b(\alpha,{\bsigma})}{\nabla b(\alpha,{\bsigma}')}\right|\leq 1 + \dfrac{\frac{12}{n}\sum_{i:i\neq j}d_{i}\bmJ_{n}(i,j)}{\Delta}
        \leq 1+\dfrac{\varepsilon}{2}\leq e^{\varepsilon/2}
    \end{align*}
    completing the proof.

\subsection{Proof of Lemma \ref{lemma:bddensityratio}}\label{sec:prfdenratio}
First suppose that we are using Algorithm \ref{PrIsing} for $(\vep,\delta)$ privacy. Let $\Gamma = b(\alpha,{\bsigma}) - b(\alpha,{\bsigma}')$. By Lemma \ref{mdiff-bound}, $|\Gamma|\leq\zeta =8 \max_{j}\frac{d_{j}}{n}$. Then, 
    \begin{align}
        \dfrac{\nu(b(\alpha,{\bsigma});\varepsilon,\delta)}{\nu(b(\alpha,{\bsigma}');\varepsilon,\delta)} &= \exp\left(\dfrac{1}{2\gamma^2}(b(\alpha,{\bsigma}')^2 - b(\alpha,{\bsigma})^2)\right) = \exp\left(\dfrac{1}{2\gamma^2}((b(\alpha,{\bsigma})-\Gamma)^2 - b(\alpha,{\bsigma})^2)\right)\nonumber \\
        &= \exp\left(\dfrac{1}{2\gamma^2}(-2b(\alpha,{\bsigma})\Gamma) + \Gamma^2)\right)\nonumber \\
        &\le \exp\left(\dfrac{1}{2\gamma^2}|2b(\alpha,{\bsigma})|\zeta + \zeta^2\right) \label{ratiobound1}
    \end{align}
    Note that for any random variable $Z\sim\mathcal{N}(0,1)$
    \begin{align*}
        \P(|Z|>t)\le 2e^{-t^2/2},\text{ for }t>1.
    \end{align*}
    Hence for, $b(\alpha,{\bsigma})\sim\mathcal{N}(0, \gamma^2)$,
    \begin{align*}
        \P(|b(\alpha,{\bsigma})|\geq \gamma t)\le 2e^{-t^2/2},\text{ for any }t>1.
    \end{align*}
    Let $S_{t}:= \{a\in\R: |a|\ge \gamma t\}$. Then it is easy to observe that for $\delta<\frac{2}{\sqrt{e}}$, and choosing $t_{0} = \sqrt{2\log(2/\delta)}$ we get,
    \begin{align*}
        \P\left(b(\alpha,{\bsigma})\in S_{t_{0}}\right)\leq \delta
    \end{align*}
    Thus on the set $S:= S_{t_{0}}^{c}$, using \eqref{ratiobound1} and recalling the definition of $\gamma$ from Algorithm \ref{PrIsing} we find,
    \begin{align*}
        \dfrac{\nu(b(\alpha,{\bsigma});\varepsilon,\delta)}{\nu(b(\alpha,{\bsigma}');\varepsilon,\delta)}\leq \exp\left(\dfrac{1}{2\gamma^2}\left\{\gamma\zeta\sqrt{8\log\dfrac 2\delta} +\zeta^2\right\}\right)\leq e^{\vep/2}
    \end{align*}
    which completes the proof of Lemma \ref{lemma:bddensityratio} for $(\vep,\delta)$ privacy. Now suppose we are using Algorithm \ref{PrIsing} for $(\vep,0)$ privacy. Then by definition,
    \begin{align*}
        \dfrac{\nu(b(\alpha,{\bsigma});\varepsilon,0)}{\nu(b(\alpha,{\bsigma}');\varepsilon,0)}
        = \exp\left(\frac{\vep}{2\zeta}\left(|b(\alpha,{\bsigma}')| - |b(\alpha,{\bsigma})|\right)\right)\leq \exp\left(\frac{\vep}{2\zeta}\left|b(\alpha,{\bsigma}') - b(\alpha,{\bsigma})\right|\right)\leq \exp\left(\frac{\vep}{2}\right)
    \end{align*}
    where the upper bound once again follows from Lemma \ref{mdiff-bound}.


\section{Regret Bound of \texttt{PrIsing} Algorithm}\label{sec:proofofupperbdd}

In this section we embark on a careful analysis of the \texttt{PrIsing} Algorithm and provide a detailed regret bound on the performance of the same. The performance of the non-private MPLE was analysed Theorem 2.1 from \cite{bhattacharya2018inference}, where the authors concluded that the estimator is $\sqrt{a_{n}}$ consistent, where, under certain regularity conditions on the log-partition function, $a_{n}$ is the Frobenius norm of the matrix $\bm J_{n}$. In the following result we recall the sufficient conditions for consistency of MPLE from \cite{bhattacharya2018inference}, and analyze the performance of \texttt{PrIsing} under the same.

\begin{theorem} \label{thm:upperbdddetailed}
    Let $\sup_{n\geq 1}\|\bm J_{n}\|<\infty$, and let $\beta_{0}>0$ be fixed. Suppose $\{a_{n}:n\geq 1\}$ is a sequence such that,
    \begin{align*}
        a_{n}\overset{n\rightarrow\infty}{\longrightarrow}\infty
    \end{align*}
    and for some $\vartheta>0$,
    \begin{align*}
        0<\liminf_{n\rightarrow\infty}\frac{1}{a_{n}}F_{n}(\beta_{0} - \vartheta)\leq \limsup_{n\rightarrow\infty}\frac{1}{a_{n}}F_{n}(\beta_{0} + \vartheta)<\infty.
    \end{align*}
    Further assume that,
    \begin{itemize}
        \item[(i)]$u_{n,K}:= \E_{\beta_{0}}\left[\sum_{i=1}^{n}|m_{i}(\bm\sigma)|\one\left\{|m_{i}(\bm\sigma)|>K\right\}\right]$
        is such that $\limsup_{K\rightarrow\infty}\limsup_{n\rightarrow\infty}\frac{1}{a_{n}} u_{n,K} = 0$, and 
        \item[(ii)] $\limsup_{n\rightarrow\infty}\frac{1}{a_{n}}\sum_{i,j=1}^{n}\bm J_{n}(i,j)^2<\infty$
    \end{itemize}
    Then the private MPLE estimator $\hat{\beta}^{\p}$ from Algorithm \ref{PrIsing} satisfies,
    \begin{align*}
        |\widehat{\beta}^\p - \beta_{0}| &= O_{p}\left({\frac{1}{\sqrt{a_{n}}} + \frac{\sqrt{8\zeta^2\rho_{\vep,\delta} + \vep^2\Delta^2\beta_{0}^2}}{a_{n}\varepsilon}}\right),
    \end{align*}
    where $$\rho_{\vep,\delta}=\begin{cases} \log (2/\delta) + \vep/2 &\text{ if } \delta >0 \\ 1 &\text{ if }\delta = 0\end{cases}$$
\end{theorem}

\begin{proof}

We prove Theorem \ref{thm:upperbdddetailed} by following the techniques developed in \cite{bhattacharya2018inference}. For notational convinience define,
    \begin{align}\label{eq:defkndelta}
        k_{n,\delta}: = \dfrac{n^{2}\varepsilon^{2}}{\zeta^{2}\left(8\log(2/\delta) + 4\varepsilon\right) + \vep^2\Delta^2\beta_{0}^2} \text{ for all } \delta>0,\ k_{n,0}:= \dfrac{n^2\vep^2}{8\zeta^2 + \vep^2\Delta^2\beta_{0}^2}.
    \end{align}
    and consider,
    \begin{align*}
        s_{n,\delta}^{2}:= \dfrac{n^{2}k_{n,\delta}}{\left(\sqrt{a_{n}k_{n,\delta}} + n\right)^2} = 
        \begin{cases}
            \dfrac{1}{\frac{1}{\sqrt{a_{n}}} + \frac{\sqrt{\zeta^2(8\log(2/\delta) + 4\varepsilon) + \vep^2\Delta^2\beta_{0}^2}}{a_{n}\varepsilon}} & \text{ if }\delta>0\\
            \dfrac{1}{\frac{1}{\sqrt{a_{n}}} + \frac{\sqrt{8\zeta^2 + \vep^2\Delta^2\beta_{0}^2}}{a_{n}\varepsilon}}& \text{ if }\delta = 0
        \end{cases}.
    \end{align*}
    Further we will use $(\vep, 0)$ privacy in place of $\vep$-privacy for consistency in notation and $\lesssim_{\theta}$ to denote less than equality upto a constant depending on a parameter $\theta$.\\
    By definition it is easy to observe that,
    \begin{align}\label{eq:bddsn2}
        s_{n,\delta}^{2}\leq \min\left\{\frac{n^{2}}{a_{n}}, k_{n,\delta}.\right\}
    \end{align}
    Recall that by Algorithm \ref{PrIsing}, $\widehat{\beta}^{\texttt{priv}}$ is the solution to the equation,
    \begin{align*}
        \mathcal{L}_{\bsigma}(\beta,b) := L_{\bsigma}(\beta) + \frac{\Delta}{n}\beta + \frac{b}{n} = 0.
    \end{align*}
    Then for $\delta>0$ with $(\vep,\delta)$ privacy,
    \begin{align}\label{eq:L2-Lsigma-finite}
        \limsup_{n\rightarrow\infty}s_{n,\delta}^{2}\mathbb{E}_{\beta_{0}, b\sim N(0,\gamma^{2})}\left[\mathcal{L}_{\bsigma}(\beta_{0},b)^{2}\right]
        &\lesssim \limsup_{n\rightarrow\infty}s_{n,\delta}^{2}\mathbb{E}_{\beta_{0}}L_{\bsigma}(\beta_{0})^{2} + \frac{s_{n,\delta}^{2}}{n^{2}}\Delta^{2}\beta_{0}^{2} + \frac{s_{n,\delta}^{2}}{n^{2}}\mathbb{E}_{\mathrm{N}(0,\gamma^{2})}b^{2}\nonumber\\
        &\leq \limsup_{n\rightarrow\infty} \frac{n^{2}}{a_{n}}\mathbb{E}_{\beta_{0}}L_{\bsigma}(\beta_{0})^2 + 1 <\infty.
    \end{align}
    where the finiteness follows by \cite[Lemma 5.2]{bhattacharya2018inference}, the definition of $k_{n,\delta}$ from \eqref{eq:defkndelta}, $\gamma, \Delta$ from Algorithm \ref{PrIsing} and observing that,
    \begin{align}\label{eq:epsilondeltaprivbdd}
        \frac{s_{n,\delta}^{2}}{n^{2}}\Delta^{2}\beta_{0}^{2} + \frac{s_{n,\delta}^{2}}{n^{2}}\gamma^{2}\leq \frac{k_{n,\delta}}{n^{2}}\Delta^{2}\beta_{0}^{2} + \frac{k_{n,\delta}}{n^{2}}\gamma^{2}\leq \frac{\varepsilon^{2}\Delta^{2}\beta_{0}^2 + \zeta^{2}(8\log(2/\delta) + 4\varepsilon)}{\zeta^{2}\left(8\log(2/\delta) + 4\varepsilon\right) + \vep^2\Delta^2\beta_{0}^2}\leq 1.
    \end{align}
    where the first inequality follows from \eqref{eq:bddsn2}. Note that for $b\sim \mathrm{Lap}(0,2\zeta/\vep)$, $ \E\left[b^2\right] = 8(\zeta/\vep)^2.$ Then for $\delta = 0$ by a similar computation as in  \eqref{eq:L2-Lsigma-finite} and \eqref{eq:epsilondeltaprivbdd} we get,
    \begin{align}\label{eq:L2laplacefinite}
        \limsup_{n\rightarrow\infty}s_{n,\delta}^{2}\mathbb{E}_{\beta_{0}, b\sim \mathrm{Lap}(0,2\zeta/\vep)}\left[\mathcal{L}_{\bsigma}(\beta_{0},b)^{2}\right] < \limsup_{n\rightarrow\infty} \frac{n^{2}}{a_{n}}\mathbb{E}_{\beta_{0}}L_{\bsigma}(\beta_{0})^2 + 1<\infty.
    \end{align} 
    Fix $\delta\geq 0$ and fix $\xi>0$, then by Chebyshev inequality, \eqref{eq:L2-Lsigma-finite} and \eqref{eq:L2laplacefinite} we can choose $K_{1} = K_{1}(\xi)>0$ such that,
    \begin{align}\label{eq:Lemma5.2-equiv}
        \mathbb{P}\left(\left|\mathcal{L}_{\bsigma}(\beta_{0},b)\right|>\frac{K_{1}}{s_{n,\delta}}\right)\leq \frac{s_{n,\delta}^{2}}{K_{1}^2}\mathbb{E}\mathcal{L}_{\bsigma}(\beta_{0},b)^{2}\lesssim_{\beta_{0}}\frac{1}{K_{1}^2}< \xi.
    \end{align}
    By \cite[Lemma 5.3]{bhattacharya2018inference} it is easy to observe that there exists $\nu:=\nu(\xi)$ and $K_{2} = K_{2}(\nu,\xi)$ such that,
    \begin{align}\label{eq:Lemma5.3-equiv}
        \mathbb{P}_{\beta_{0}}\left(\sum_{i=1}^{n}m_{i}(\bsigma)^{2}\bm{1}\left\{|m_{i}(\bsigma)|\leq K_{2}\right\}\geq \nu a_{n}\right)\geq 1-\xi
    \end{align}
    for large enough $n$.
    Define,
    \begin{align*}
        T_{n}:= \left\{(\bsigma,b)\in \{+1,-1\}^{n}\times \mathbb{R}: |\mathcal{L}_{\bsigma}(\beta_{0},b)|\leq \frac{K_{1}}{s_{n,\delta}}, \sum_{i=1}^{n}m_{i}(\bsigma)^{2}\bm{1}\left\{|m_{i}(\bsigma)|\leq K_{2}\right\}\geq \nu a_{n}\right\}.
    \end{align*}
    Combining \eqref{eq:Lemma5.2-equiv} and \eqref{eq:Lemma5.3-equiv} and taking $n$ large enough we conclude that,
    \begin{align*}
        \mathbb{P}(T_{n})\geq 1-2\xi.
    \end{align*}
    Now choosing $(\bsigma,b)\in T_{n}$ and recalling that the parameter $\beta\geq 0$ shows,
    \begin{align}
        \mathcal{L}_{\bsigma}'(\beta,b):=\frac{\partial}{\partial\beta}\mathcal{L}_{\bsigma}(\beta,b) 
        & = \frac{1}{n}\sum_{i=1}^{n}m_{i}(\bsigma)^{2}\sech^{2}(\beta m_{i}(\bsigma)) + \frac{\Delta}{n}\nonumber\\
        & \geq \frac{1}{n}\sech^{2}(\beta K_{2})\sum_{i}^{n}m_{i}(\bsigma)^{2}\bm{1}\left\{|m_{i}(\bsigma)|\leq K_{2}\right\} + \frac{\Delta}{n}\nonumber\\
        & \geq \nu \frac{a_{n}}{n}\sech^{2}(\beta K_{2}) + \frac{\Delta}{n}.\label{eq:lowbddderivative}
    \end{align}
    Thus,
    \begin{align}
        \frac{K_{1}}{s_{n,\delta}}\geq |\mathcal{L}_{\bsigma}(\beta_{0},b)| 
        & = |\mathcal{L}_{\bsigma}(\beta_{0},b) - \mathcal{L}_{\bsigma}(\widehat{\beta}^{\texttt{priv}},b)|\nonumber\\
        & \geq \int_{\widehat{\beta}^{\texttt{priv}}\wedge \beta_{0}}^{\widehat{\beta}^{\texttt{priv}}\vee \beta_{0}}\mathcal{L}_{\sigma}'(\beta,b)d\beta\nonumber\\
        & \geq \left|\nu\frac{a_{n}}{K_{2}n}\tanh(K_{2}\widehat{\beta}^{\texttt{priv}}) + \frac{\Delta}{n}\widehat{\beta}^{\texttt{priv}} - \nu\frac{a_{n}}{K_{2}n}\tanh(K_{2}\beta_{0}) - \frac{\Delta}{n}\beta_{0}\right|\label{eq:lbbderapply}\\
        & = \nu\frac{a_{n}}{K_{2}n}\left|\tanh(K_{2}\widehat{\beta}^{\texttt{priv}}) + \frac{K_{2}\Delta}{\nu a_{n}}\widehat{\beta}^{\texttt{priv}} - \tanh(K_{2}\beta_{0}) - \frac{K_{2}\Delta}{\nu a_{n}}\beta_{0}\right|.\nonumber
    \end{align}
    where the inequality in \eqref{eq:lbbderapply} follows from \eqref{eq:lowbddderivative}. Now recalling our choice of $(\bsigma,b)\in T_{n}$ we conclude,
    \begin{align*}
        \mathbb{P}\left(\frac{a_{n}s_{n,\delta}}{n}\left|\tanh(K_{2}\widehat{\beta}^{\texttt{priv}}) + \frac{K_{2}\Delta}{\nu a_{n}}\widehat{\beta}^{\texttt{priv}} - \tanh(K_{2}\beta_{0}) - \frac{K_{2}\Delta}{\nu a_{n}}\beta_{0}\right|\geq\frac{K_{2}}{\nu K_{1}} \right)\leq 2\xi
    \end{align*}
    for large enough $n$. The proof is now complete by invoking Lemma \ref{lemma:tanhOp1}.

\end{proof}

\subsection{Proof of Theorem \ref{upperbound}}
Note that all the assumptions of Theorem \ref{thm:upperbdddetailed} are satisfied. By Algorithm \ref{PrIsing} the smallest permitted value of $\Delta$ is given by,
\begin{align}\label{eq:defDelta0}
    \Delta_{0} = \max_{j}\left\{\frac{24}{\vep n}\sum_{i=1}^{n}d_{i}\bmJ_{n}(i,j)\right\}.
\end{align}
Fix $1\leq j\leq n$. By the definition of $d_{i}, 1\leq i\leq n$ from Algorithm \ref{PrIsing} note that,
\begin{align*}
    \sum_{i=1}^{n}d_{i}\bmJ_{n}(i,j) = n\sum_{i=1}^{n}\sum_{k=1}^{n} \bmJ_{n}(i,k)\bmJ_{n}(i,j) = n\sum_{k=1}^{n}\sum_{i=1}^{n}\bmJ_{n}(k,i)\bmJ_{n}(i,j) = n\sum_{k=1}^{n}\bmJ_{n}^2(k,j).
\end{align*}
By \eqref{eq:defDelta0} note that,
\begin{align}\label{eq:formDelta0}
    \vep\Delta_{0} = 24\max_{j}\left\{\sum_{k=1}^{n}\bmJ_{n}^2(k,j)\right\} = 24\left\|\bm{J}_{n}^2\right\|_{1\rightarrow\infty}.
\end{align}
Now for $\zeta$ from Algorithm \ref{PrIsing} we have,
\begin{align}\label{eq:formzeta}
    \zeta = 8\max_{i}\left\{\frac{d_{i}}{n}\right\} = 8\max_{i}\left\{\sum_{j=1}^{n}\bmJ_{n}(i,j)\right\} = 8\max_{j}\left\{\sum_{i=1}^{n}\bmJ_{n}(i,j)\right\} = 8\left\|\bmJ_{n}\right\|_{1\rightarrow\infty}
\end{align}
where the penultimate equality follows since $\bmJ_{n}$ is symmetric. Now recalling $\rho_{\vep,\delta}$ from Theorem \ref{thm:upperbdddetailed}, \eqref{eq:formDelta0} and \eqref{eq:formzeta} shows,
\begin{align*}
    8\zeta^2\rho_{\vep,\delta} + \vep^2\Delta_{0}^2\beta_{0}^2 \lesssim \left\|\bmJ_{n}\right\|_{1\rightarrow\infty}^2\rho_{\vep,\delta} + \left\|\bm{J}_{n}^2\right\|_{1\rightarrow\infty}^2\beta_{0}^2\leq \max\left\{1, \left\|\bm{J}_{n}\right\|_{1\rightarrow\infty}^4\right\}\left(\rho_{\vep,\delta} + \beta_{0}^2\right).
\end{align*}
The result now follows from Theorem \ref{thm:upperbdddetailed}.

\section{Additional Experiments}
We report additional simulations to evaluate the cost of privacy. In Figure \ref{allinone} we generate Ising model synthetic outcomes on Erd\H{o}s-R\'enyi, HIV network and the Political Blogs network, and plot the MSE of the private-MPLE estimates across a wide range of $\varepsilon$, and in the regimes of $\beta>1$ and $\beta<1$. The cost, quantified by the MSE shows a decreasing trend with $\varepsilon$, with the difference in regimes being stark in the Erd\H{o}s-R\'enyi network.
\begin{figure}[!ht]
    \centering
    \includegraphics*[width = 0.45\textwidth]{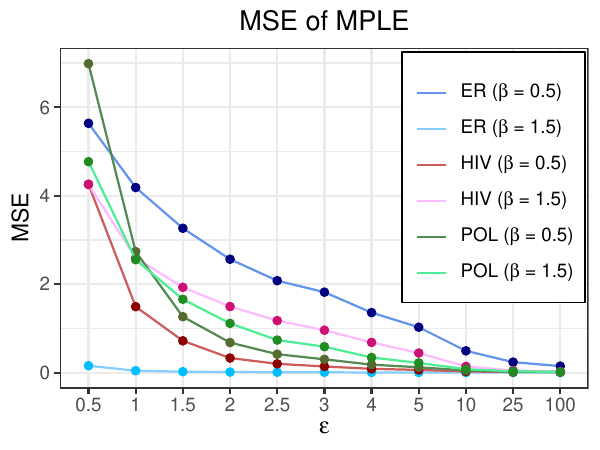}
    \caption{MSE of \texttt{PrIsing} estimates across $\varepsilon$ for all networks in the paper}
    \label{allinone}
\end{figure}

Next, we compare the privacy costs in a neighborhood of the estimated $\hat\beta$s. As noted in Sections \ref{HIVsubsect} and Section \ref{sec:Polblogs} corresponding beta-hat turns out to be 1.8 and 2.85 respectively. As before, we generate Ising model synthetic outcomes with beta in a range around $\hat\beta$, and estimate $\hat{\beta}^{\p}$ $500$ times to produce MSE values. We plot the results varying across epsilon, and plot the results in Figure \ref{fig:S34}(a) for HIV network and Figure \ref{fig:S34}(b) for political blogs network. 

\begin{figure}[!ht]
    \centering
    \begin{subfigure}[c]{0.45\textwidth}
       \includegraphics[width=1\linewidth]{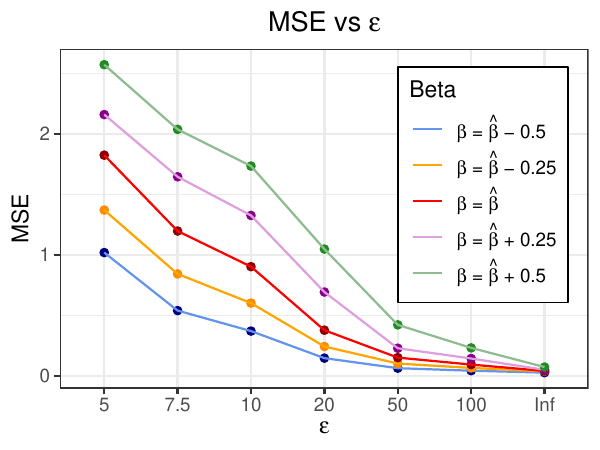} 
          \caption*{\small{(a)}} 
    \end{subfigure} 
    \begin{subfigure}[c]{0.45\textwidth}
       \includegraphics[width=1\linewidth]{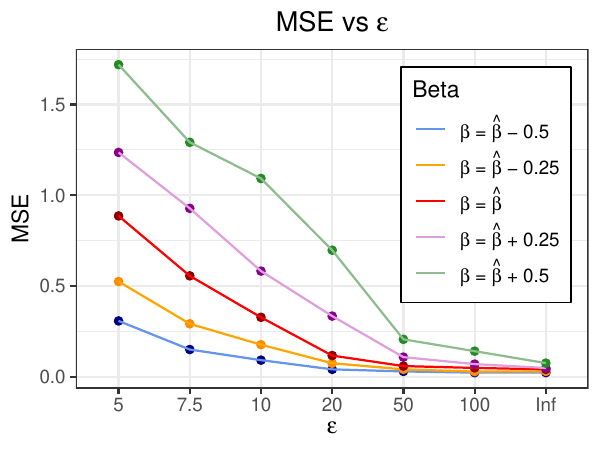} 
          \caption*{\small{(b)}} 
    \end{subfigure} 
    \caption{MSE of \texttt{PrIsing} estimates across a range of $\beta$ values around $\hat\beta$ using Ising Models on (a) HIV network and (b) Political Blogs network.}
    \label{fig:S34}
\end{figure} 

The results show a decreasing trend in MSE with increasing epsilon, re-ensuring that the MSE decreases as the privacy guarantee becomes weaker.

\section{Technical Details}
In this section we provide technical lemmas needed for completing the proof of our results.
\begin{lemma}\label{lemma:tanhOp1}
Consider a sequence of random variables $\{X_{n}:n\geq 1\}$ and suppose that for every $\xi>0$ there exists $K_{1}(\xi), K_{2}(\xi), K_{3}(\xi)>0$ such that, 
\begin{align}\label{eq:Mnconvg}
    \P\left(M_{n}\left|\tanh(K_{1}(\xi)X_{n}) + K_{2}(\xi)t_{n}X_{n} - \tanh(K_{1}(\xi)c) - K_{2}(\xi)t_{n}c\right|>K_{3}(\xi)\right)\leq \xi
\end{align}
for all $n\geq n_{0}(\xi)$, where $c>0$ is a constant, $t_{n}>0\ \forall n\geq 1$ and $M_{n}\rightarrow\infty$ as $n\rightarrow\infty$. Then,
\begin{align*}
    M_{n}\left|X_{n} - c\right| = O_{p}(1)
\end{align*}
\begin{proof}
Observe that,
\begin{align*}
    \bigg|\tanh(K_{1}(\xi)X_{n}) + K_{2}(\xi)t_{n}X_{n}
    & - \tanh(K_{1}(\xi)c) - K_{2}(\xi)t_{n}c\bigg|\\
    & = \bigg|\tanh(K_{1}(\xi)X_{n}) - \tanh(K_{1}(\xi)c)\bigg| + \bigg|K_{2}(\xi)t_{n}X_{n} - K_{2}(\xi)t_{n}c\bigg|
\end{align*}
Then by \eqref{eq:Mnconvg} we get,
\begin{align*}
    \P\bigg(M_{n}
    &\bigg|\tanh(K_{1}(\xi)X_{n}) - \tanh(K_{1}(\xi)c)\bigg|>K_{3}(\xi)\bigg)\\
    & \leq \P\left(M_{n}\left|\tanh(K_{1}(\xi)X_{n}) + K_{2}(\xi)t_{n}X_{n} - \tanh(K_{1}(\xi)c) - K_{2}(\xi)t_{n}c\right|>K_{3}(\xi)\right)\leq \xi
\end{align*}
for all $n\geq n_{0}(\xi)$. Now for fixed $\xi>0$ and using the mean value theorem,
\begin{align}\label{eq:useMVT}
    M_{n}|X_{n} - c|
    & = \frac{M_{n}}{K_{1}(\xi)}\left|\tanh^{-1}(\tanh(K_{1}(\xi)X_{n})) - \tanh^{-1}(\tanh(K_{1}(\xi)c))\right|\nonumber\\
    &\leq \frac{M_{n}}{K_{1}(\xi)}\left|\dfrac{\tanh(K_{1}(\xi)X_{n}) - \tanh(K_{1}(\xi)c)}{1-\zeta_{\xi}^2}\right|
\end{align}
where $\min\left\{\tanh(K_{1}(\xi)X_{n}),\tanh(K_{1}(\xi)c)\right\}\leq \zeta_{\xi}\leq \max\left\{\tanh(K_{1}(\xi)X_{n}),\tanh(K_{1}(\xi)c)\right\}$. By definition,
\begin{align*}
    |1-\zeta_{\xi}^2| = 1-\zeta_{\xi}^2\geq 1-|\zeta_{\xi}|\geq 1-|\tanh(K_{1}(\xi)c)| - |\tanh(K_{1}(\xi)X_{n}) - \tanh(K_{1}(\xi)c)|
\end{align*} 
Since $M_{n}\rightarrow\infty$, then there exists $n_{1}(\xi)>n_{0}(\xi)$ such that for all $n\geq n_{1}(\xi)$, 
\begin{align}\label{eq:relK4K3}
    \dfrac{K_{3}(\xi)}{M_{n}}\leq K_{4}(\xi):=\frac{1}{2}(1-|\tanh(K_{1}(\xi)c)|)
\end{align}
Note that on the event $\left|\tanh(K_{1}(\xi)X_{n}) - \tanh(K_{1}(\xi)c)\right|\leq K_{4}(\xi)$ with \eqref{eq:relK4K3}, we have,
\begin{align*}
    |1-\zeta_{\xi}^2| \geq K_{4}(\xi).
\end{align*}
Hence recalling \eqref{eq:useMVT}, on the event $\left|\tanh(K_{1}(\xi)X_{n}) - \tanh(K_{1}(\xi)c)\right|\leq K_{4}(\xi)$ we get,
\begin{align*}
    M_{n}|X_{n} - c|\leq \frac{M_{n}}{K_{1}(\xi)K_{4}(\xi)}\left|\tanh(K_{1}(\xi)X_{n}) - \tanh(K_{1}(\xi)c)\right|
\end{align*}
Now choosing $P(\xi) = \frac{K_{3}(\xi)}{K_{1}(\xi)K_{4}(\xi)}$ shows,
\begin{align*}
    \P\left(M_{n}\left|X_{n} - c\right|>P(\xi)\right)
    &\leq \P\left(M_{n}\left|X_{n} - c\right|>P(\xi), \left|\tanh(K_{1}(\xi)X_{n}) - \tanh(K_{1}(\xi)c)\right|\leq K_{4}(\xi)\right)\\
    &\hspace{20pt} + \P\left(\left|\tanh(K_{1}(\xi)X_{n}) - \tanh(K_{1}(\xi)c)\right|> K_{4}(\xi)\right)\\
    &\leq 2\P\left(\left|\tanh(K_{1}(\xi)X_{n}) - \tanh(K_{1}(\xi)c)\right|> \frac{K_{3}(\xi)}{M_{n}}\right)\leq 2\xi
\end{align*}
for all $n\geq n_{1}(\xi)$, which completes the proof.
\end{proof}
\end{lemma}

\end{document}